\documentclass[runningheads,orivec]{llncs}

\usepackage[T1]{fontenc}
\usepackage[utf8]{inputenc}
\usepackage{amsmath,amssymb}
\usepackage{newtxtext}
\usepackage[varvw]{newtxmath}
\usepackage{url}
\usepackage[colorlinks,allcolors=black,urlcolor=blue]{hyperref}

\AtBeginDocument{\renewcommand{\doi}[1]{\url{https://doi.org/#1}}}
\hypersetup{
  pdftitle={Fast Leaf-to-Ancestor Minimum Query in the Oracle Model},
  pdfauthor={Aleksey Upirvitskiy and Aleksandr Levin}
}

\DeclareMathOperator{\pre}{pre}
\DeclareMathOperator{\Comp}{Comp}
\DeclareMathOperator{\Oracle}{Oracle}
\DeclareMathOperator{\depth}{depth}
\DeclareMathOperator{\RMQ}{RMQ}
\DeclareMathOperator{\lowerDist}{lowerDist}

\begin{document}

\title{Fast Leaf-to-Ancestor Minimum Query in the Oracle~Model}
\titlerunning{Fast Leaf-to-Ancestor Minimum Query}

\author{Aleksey Upirvitskiy\inst{1}\thanks{Corresponding author.}
\and Aleksandr Levin\inst{2}\thanks{Corresponding author.}}
\authorrunning{A. Upirvitskiy and A. Levin}

\institute{Independent researcher\\
\email{\href{mailto:aleks5d@mail.ru}{aleks5d@mail.ru}}
\and
HSE University\\
\email{\href{mailto:am.levin@hse.ru}{am.levin@hse.ru}}}

\maketitle

\begin{abstract}
We study path-minimum queries from a vertex to one of its ancestors in a
rooted tree whose edge weights are accessible only through a comparison
oracle. For a tree with $n$ vertices and height $h\ge 2$, we give a static
data structure that uses $O(n\log h)$ preprocessing time, space, and oracle
calls. It returns a minimum-weight edge on the query path in $O(1)$
worst-case time with zero oracle calls. Binary lifting and a range-minimum
structure on a ladder decomposition supply at most two candidate edges. Stored
nearest-smaller-ancestor information determines which candidate to return.
We prove a matching deterministic lower bound on preprocessing comparisons
even when the starting vertex of every query is a leaf.
The structure also resolves Yang's open question on the preprocessing and
query tradeoff for arbitrary tree paths. For each fixed $j\ge 1$, the
resulting structure uses $O(n\log\lambda_j(n))$ preprocessing time and space,
at most $2j$ comparisons per query, and $O(j)$ worst-case query time, where
$\lambda_j$ is the inverse of the $j$th row of Yang's Ackermann hierarchy.

\end{abstract}

\section{Introduction}

Path extrema arise in minimum spanning tree (MST) verification. Given a
candidate tree $T$ and a new edge $e=(u,v)$, the tree is minimum in
$T\cup\{e\}$ exactly when the weight of $e$ is at least the maximum edge
weight on the tree path from $u$ to $v$. For a fixed $T$, one can preprocess
its weights before the edges to be tested arrive. If an endpoint is an
ancestor of the other, a stored query structure can identify a maximum edge
on that path, leaving just the comparison with $e$ for verification. Reversing
the weight order exchanges maxima and minima. We ask how much preprocessing
is needed to answer such upward path queries when all comparisons between
tree-edge weights must take place during preprocessing.

Yang~\cite[Sect.~4]{yang2021} asks whether arbitrary tree paths admit
$O(n\log\lambda_j(n))$ preprocessing and $O(j)$ query time using at most
$2j$ query comparisons. He identifies a sufficient primitive for
leaf-to-ancestor paths: $O(n\log h)$ preprocessing, constant query time,
and zero query comparisons. Theorem~\ref{thm:main} supplies this primitive,
and Corollary~\ref{cor:yang} gives the resulting arbitrary-path tradeoff.

\subsection{Model and Main Result}

The input is the topology of a rooted tree $T$ with $n$ vertices, together with
an oracle $\Oracle(e,f)$ that returns true exactly when the weight of edge $e$
is strictly less than that of edge $f$. The weights themselves are hidden.
The depth $\depth(v)$ is the number of edges from the root to $v$, and the
height is $h=\max_{v\in V(T)}\depth(v)$. A query gives two vertex identifiers
$(v,u)$, with $u$ an ancestor of $v$, and asks for the identifier of a
minimum-weight edge on their path. If $v=u$, the answer is $\mathrm{EMPTY}$.
In the leaf-to-ancestor variant, $v$ is required to be a leaf. Queries arrive
one at a time after preprocessing.

We use the cost triple $(P,k,Q)$, where $P$ is preprocessing time with
unit-cost oracle calls, $k$ is the maximum number of oracle calls per query,
and $Q$ is the worst-case number of additional word-RAM operations per query.
Space is measured in words of $\Theta(\log(n+1))$ bits, sufficient for vertex
identifiers and table addresses. We also count preprocessing oracle calls
separately.

\begin{theorem}\label{thm:main}
Let $T$ have $n$ vertices and height $h\ge 2$. Descendant-to-ancestor
minimum queries admit a deterministic data structure with
\[
(P,k,Q)=(O(n\log h),0,O(1)).
\]
The structure uses $O(n\log h)$ words and preprocessing oracle calls.
For \mbox{$2\le h\le n-1$}, a deterministic structure with zero query
calls requires $\Omega(n\log h)$ preprocessing comparisons in the worst
case on $n$-vertex trees of height at most $h$, even for leaf-to-ancestor queries.
\end{theorem}

For $h\le 1$, each query path has at most one edge and can be answered in
constant time with $O(n)$ preprocessing and no oracle calls. For larger
heights, a binary-lifting table supplies the minimum of a power-of-two prefix.
If the remainder is nonempty, a range-minimum structure on a ladder supplies
its minimum. Since the two candidates lie on the same upward path, the
distance to the nearest smaller ancestor suffices to select the answer.
Sections~\ref{sec:algorithm} and~\ref{sec:analysis} give the construction and its
costs. Section~\ref{sec:lower-bound} proves the lower bound.
For arbitrary leaf-to-leaf queries, zero query calls can force
$\Omega(n\log n)$ preprocessing comparisons even at height two.

\subsection{Related Work}

\paragraph{Offline and Upward Path Queries.}
Koml\'os~\cite{komlos1985} bounds the number of comparisons for computing the
maxima of $m$ prescribed upward paths by $O(n+n\log(1+m/n))$. This gives a
linear comparison bound for offline MST verification. For online
leaf-to-ancestor queries, Yang~\cite[Sect.~3.5]{yang2021} implements
Koml\'os's stack method with a persistent balanced search tree. That
structure uses $O(n\log h)$ preprocessing time and space and answers
queries in $O(\log h)$ time with zero oracle calls. Theorem~\ref{thm:main}
reduces the query time to $O(1)$.

\paragraph{Online Verification.}
Pettie~\cite{pettie2006} gives inverse-Ackermann tradeoffs for online MST
verification. For a fixed budget of $t\ge 1$ comparisons per verification
query, his lower bound requires $\Omega(n\log\lambda_t^{\mathrm P}(n))$
comparisons before a hard query, where $\lambda_t^{\mathrm P}$ is the
$t$th row inverse in his Ackermann hierarchy. The hard instances use
leaf-to-ancestor query edges on balanced binary trees. An upward path maximum
structure with budget $k$ gives a verifier for these queries with budget $k+1$,
including the comparison with the new edge. In particular, the lower bound
for $t=1$ gives $\Omega(n\log\log n)$ preprocessing comparisons for
zero-call upward queries on these logarithmic-height trees. For
verification with arbitrary endpoints, Pettie~\cite[Thm.~4.2]{pettie2006}
gives optimal preprocessing bounds of $\Theta(n\log n)$ with one
comparison per query and $\Theta(n\log\log n)$ with two.

\paragraph{Arbitrary Paths.}
King's Bor\r{u}vka-tree construction~\cite{king1997} preserves path extrema
on paths between corresponding leaves of an auxiliary tree. Yang
combines this transformation with a recursive decomposition by depth.
Table~\ref{tab:tradeoffs} compares the resulting bounds.
Here $\lambda_j$ is the inverse of the $j$th row of Yang's Ackermann
hierarchy~\cite[Def.~5]{yang2021}, with $j\ge 1$ in the tradeoff bounds. The expression
$\log\lambda_j(n)$ is truncated below by one.

Alon and Schieber~\cite{alon1987} and Chazelle~\cite{chazelle1987} study
online tree-path products over a semigroup. Kahalon et al.~\cite{kahalon2022} use
bounded-hop $1$-spanners of tree metrics to find shortcut routes in time
proportional to the hop bound. In the unit-length metric of the input tree,
a route of $s$ hops partitions the query path into $s$ subpaths. If their minima
are stored, the query can combine the $s$ candidates using $s-1$ comparisons.

\begin{table}[t]
\caption{Preprocessing and query bounds. The column $k$ gives the maximum
number of oracle calls per query. The starred budget includes the
comparison with the arriving edge.}
\label{tab:tradeoffs}
\centering\small
\setlength{\tabcolsep}{3pt}
\begin{tabular}{@{}p{.37\linewidth}lcl@{}}
\hline
Problem and source & $P$ & $k$ & $Q$\\
\hline
Leaf to ancestor \cite[Sect.~3.5]{yang2021}
 & $O(n\log h)$ & $0$ & $O(\log h)$\\
Descendant to ancestor (Thm.~\ref{thm:main})
 & $O(n\log h)$ & $0$ & $O(1)$\\
Arbitrary path \cite[Sect.~2.2]{yang2021}
 & $O(n\log n)$ & $0$ & $O(1)$\\
Arbitrary path \cite[Sect.~3.5]{yang2021}
 & $O(n\log\lambda_j(n))$ & $2j$ & $O(j+\log\lambda_j(n))$\\
Arbitrary path \cite[Sect.~3.5]{yang2021}
 & $O(n\log\lambda_j(n))$ & $2j+2$ & $O(j)$\\
Arbitrary path (Cor.~\ref{cor:yang})
 & $O(n\log\lambda_j(n))$ & $2j$ & $O(j)$\\
Online verification (Cor.~\ref{cor:verification})
 & $O(n\log\log n)$ & $2^{\ast}$ & $O(1)$\\
\hline
\end{tabular}
\end{table}

\section{Preliminaries}

\subsection{Ladder Decomposition}

A ladder decomposition is a tree preprocessing technique used to accelerate
queries on ancestor-directed paths and is closely related to level-ancestor
structures~\cite{bender2004}. We use a ladder to store a contiguous part of
an ancestor path in an array.

The construction consists of two parts. First, the vertices are partitioned
into disjoint base paths by a longest-path decomposition. Then each base path is
extended upward by adding ancestors of its highest vertex, until either its length
is doubled or the root is reached. Thus base paths form a partition of the tree,
whereas the resulting ladders may overlap in the added ancestor vertices. This
extension gives, for every vertex of a base path, a contiguous array of ancestors
that can be used for fast upward jumps.

\subsubsection{Building the Longest-Path Decomposition.}
The longest-path decomposition
breaks the tree into disjoint longest-path layers. It can be built in linear time as
follows:

\begin{enumerate}
\item For each node, compute the maximum distance from that node down to a leaf:
it is $0$ for a leaf, and $1+\max\{\mathrm{down}(\mathrm{child})\}$ otherwise.
\item For every non-leaf node, choose one child that attains the maximum
$\mathrm{down}$ value,
and keep only the edge to that child.
\item After removing all other child edges, the tree splits into a forest of paths.
\end{enumerate}

These paths form the longest-path decomposition. In the data structures below
we store every path from its deepest node toward the root. Thus a path can be
represented by the deepest node and by the number of vertices in the path:
\[
(v,\ell)\longrightarrow [\,v,\pre_1(v),\pre_2(v),\ldots,\pre_{\ell-1}(v)\,],
\]
where $\pre_i(v)$ denotes the ancestor of $v$ at distance $i$ ($\pre_1(v)$ is the parent of $v$,
$\pre_2(v)$ is the parent of $\pre_1(v)$, etc.).

\subsubsection{Building the Ladder Decomposition.}
Building the ladder decomposition can
be done in linear time by appending some nodes to every path of the longest
path decomposition. This updates each path descriptor by the rule
\[
(v,\ell)\to (v,\min(2\ell,\depth(v)+1)).
\]
The term $\depth(v)+1$ is the number of vertices on the path from the root to
$v$, so it is the largest possible length of a ladder ending at $v$.
The ladder obtained from a longest-path component is called the base ladder of
the vertices of that component. Each vertex stores a pointer to its base
ladder, and each ladder stores its deepest vertex.

The base paths partition the vertices, so their total length is $n$.
Extending them takes $O(n)$ time and creates at most $2n$ array entries.

\begin{lemma}\label{lem:ladder}
Let $x$ be a node with a descendant at distance $D$. Then, for every ancestor
$y$ of $x$ with $\depth(x)-\depth(y)\le D$, the base ladder of $x$ contains the
whole path from $x$ to $y$, including both endpoints.
\end{lemma}

\begin{proof}
Since $x$ has a descendant at distance $D$, the longest downward path from $x$
has length at least $D$. Therefore the path of the longest-path decomposition
that contains $x$ has at least $D$ vertices strictly below $x$. Let this
decomposition path have $S$ vertices and deepest vertex $z$. The distance from
$z$ to $y$ is at most $(S-1)+D\le 2S-2$. After the path is doubled, unless the
root is reached earlier, the base ladder contains all vertices up to this
distance from $z$. Thus the whole path from $x$ to $y$ lies in this ladder.
\end{proof}

\subsection{Range Minimum Query}

A range-minimum query (RMQ) asks for the minimum of a nonempty interval
$[L,R)$ in an array $A$. A sparse table on an array of length $m$ uses
$O(m\log(m+1))$ preprocessing time and space and answers queries in $O(1)$ time.

Define
\[
st[t][i]=\min_{j=0}^{2^t-1} A[i+j],
\]
the minimum in the block of length $2^t$ starting at $i$. We fill the table bottom-up:
\[
st[0][i]=A[i],\qquad
st[t][i]=\min\bigl(st[t-1][i],\ st[t-1][i+2^{t-1}]\bigr).
\]
To answer a query on the half-open interval $[L,R)$ we choose
\[
t=\left\lfloor \log_2(R-L)\right\rfloor,
\]
take the two possibly overlapping blocks of size $2^t$ that cover the range,
and return
\[
\RMQ(L,R)=\min\bigl(st[t][L],\ st[t][R-2^t]\bigr).
\]

\subsection{Binary Lifting}

Binary lifting stores jumps of power-of-two lengths. For every vertex $v$
and every $t\ge 0$ with $2^t\le\depth(v)$, we store the ancestor
$\mathrm{up}[t][v]=\pre_{2^t}(v)$, using the recurrences
\[
\mathrm{up}[0][v]=\mathrm{parent}(v),\qquad
\mathrm{up}[t][v]=\mathrm{up}[t-1][\,\mathrm{up}[t-1][v] \,].
\]
The first recurrence applies to non-root vertices and the second to $t\ge 1$.
Jumps with $2^t>\depth(v)$ are marked absent.

For each valid jump we store the identifier $\mathrm{minBlock}[t][v]$ of the
minimum-weight vertex on $[v,\pre_{2^t}(v))$, the upward path that excludes
its upper endpoint. The base case is
\[
\mathrm{minBlock}[0][v]=v,
\]
and
\[
\mathrm{minBlock}[t][v]
=\min\bigl(\mathrm{minBlock}[t-1][v],\
\mathrm{minBlock}[t-1][\mathrm{up}[t-1][v]]\bigr),
\]
where $\Comp$, defined below, selects the minimum during preprocessing.

\subsection{From Edge Weights to Vertex Weights}

For each non-root vertex $x$, let $e(x)=(\mathrm{parent}(x),x)$ be its incoming
edge. We use $x$ as an identifier for the weight of $e(x)$. A nonempty
edge-minimum query from $v$ to its ancestor $u$ then asks for the minimum on
the vertex path
\[
[v,u)=\{v,\mathrm{parent}(v),\ldots,\text{the child of }u\}.
\]
The answer is the incoming edge of the selected vertex. The vertex $u$ is
excluded because its incoming edge lies above the query path.

The oracle gives only a strict comparison of edge values, and equal values may
occur. We define a strict total order on vertices by breaking weight ties by
vertex identifier. The root is a sentinel larger than every non-root vertex.
Comparisons with the root use no oracle calls. To evaluate $\Comp(v,u)$ for
non-root vertices $v$ and $u$, proceed as follows.
\begin{enumerate}
\item If $\Oracle(e(v),e(u))$ is true, return true.
\item If $\Oracle(e(u),e(v))$ is true, return false.
\item Return whether $v$ precedes $u$ in the fixed identifier order.
\end{enumerate}
Each evaluation of $\Comp$ uses at most two oracle calls. In what follows,
minimum and smaller refer to this total order.

\section{Modifications of the Basic Algorithms}

\subsection{Modified RMQ}

Each ladder array contains distinct vertices. For every sparse-table block
we store the position of its minimum, using $\Comp$ during preprocessing.
A query still needs to select between the two candidate positions. To make
this choice without an oracle call, we store the closest previous position
whose vertex is smaller. For each position $i$, define
\[
\mathrm{preLowerSet}(i)=\{\, j<i \mid \Comp(A[j],A[i])=\mathrm{true}\,\}
\]
\[
\mathrm{preLower}(i)=
\begin{cases}
\max(\mathrm{preLowerSet}(i)),& \mathrm{preLowerSet}(i)\ne \varnothing,\\
-1& \text{otherwise}.
\end{cases}
\]

\begin{lemma}
For an array of $m$ distinct vertices, $\mathrm{preLower}$ can be
computed with $O(m)$ time and oracle calls during preprocessing.
\end{lemma}

\begin{proof}
Scan from left to right, maintaining a stack of positions whose vertices
increase in the $\Comp$ order. At position $i$, pop while the top vertex is
larger than $A[i]$. For any later vertex larger than a popped vertex,
$A[i]$ is both smaller and closer. The remaining top,
if any, is $\mathrm{preLower}(i)$. Then push $i$. Each position is pushed once
and popped at most once, and each test uses at most two oracle calls.
\end{proof}

\begin{lemma}
Let $A$ contain $m$ distinct vertices. With comparisons via $\Comp$,
the RMQ structure uses $O(m\log(m+1))$ preprocessing time, words, and
oracle calls. For any nonempty interval $[L,R)$, it returns the
minimum position in $O(1)$ time with zero oracle calls.
\end{lemma}

\begin{proof}

The sparse table stores, for every power-of-two block, the position of its
minimum. All choices while building the table are made by calls to $\Comp$.
A precomputed table of integer logarithms supplies $t$ in constant time.

After preprocessing, the query on $[L,R)$ is answered as follows:
\begin{enumerate}
\item Let $t=\lfloor \log_2(R-L)\rfloor$.
\item Let $M_1$ be the stored minimum position of $[L,L+2^t)$.
\item Let $M_2$ be the stored minimum position of $[R-2^t,R)$.
\item If $M_1=M_2$, return $M_1$.
\item If $\mathrm{preLower}(M_2)<M_1$, return $M_2$.
\item Otherwise return $M_1$.
\end{enumerate}
If the two positions are distinct, then $M_1<M_2$. In the reverse order both
candidates would lie in the overlap, where the minimum is unique. Let
$j=\mathrm{preLower}(M_2)$. If $A[M_1]$ is smaller than $A[M_2]$, then $M_1$
is a previous smaller position for $M_2$, so $j\ge M_1$.

Conversely, suppose $j\ge M_1$. The position $j$ is to the left of the
right block because $A[M_2]$ is its minimum and $A[j]$ is smaller.
Since $M_1\le j<M_2$, it follows that $j$ lies in the left block.
Its minimum $A[M_1]$ is no larger than $A[j]$, with equality if $j=M_1$.
Thus $A[M_1]$ is smaller than $A[M_2]$ exactly when $j\ge M_1$, and
the query selects the correct position using only indices.
\end{proof}

\subsection{Nearest Smaller Ancestors}

For every non-root vertex $v$, store the distance to its nearest smaller
ancestor in the tie-broken order,
\[
\lowerDist(v)=
\min\bigl(\{\, i:0<i<\depth(v),\ \Comp(\pre_i(v),v)\,\}\cup\{\infty\}\bigr).
\]
We represent $\infty$ by $h+1$.

To compute this entry, start at $x=\mathrm{parent}(v)$ with $d=1$.
Scan $t$ from $\lfloor\log_2 h\rfloor$ down to zero. When the block of
length $2^t$ beginning at $x$ exists and its stored minimum is at least
$v$ in the $\Comp$ order, move $x$ up by $2^t$ and add $2^t$ to $d$.

Let $s$ be the number of consecutive non-root ancestors, starting at
$\mathrm{parent}(v)$, before the first smaller ancestor or the root.
A block can be skipped precisely when it fits in the remaining part
of that prefix. Its stored minimum detects whether it contains a
smaller vertex. The scan finds the binary expansion of $s$ from the
highest bit downward. If it ends at the root, store $\lowerDist(v)=\infty$.
Otherwise $x$ is the nearest smaller ancestor and $\lowerDist(v)=d$.
This takes $O(\log h)$ time and oracle calls per non-root vertex.

\section{Main Algorithm}\label{sec:algorithm}

\subsection{Offline Construction}

\begin{enumerate}
\item Record the parent, depth, and incoming edge of each non-root vertex.
Use $\Comp$ for weight comparisons during the remaining steps.
\item Precompute the table $\mathrm{lg}[i]=\lfloor\log_2 i\rfloor$ for
$1\le i\le h+1$, using $\mathrm{lg}[1]=0$ and
$\mathrm{lg}[i]=1+\mathrm{lg}[\lfloor i/2\rfloor]$ for $i\ge 2$.
\item Build the ladder decomposition and record the base ladder of each vertex.
\item Build the modified RMQ structure on each ladder array.
\item Build the binary-lifting tables $\mathrm{up}$ and $\mathrm{minBlock}$.
\item Compute $\lowerDist$ for the non-root vertices.
\end{enumerate}

\subsection{Querying}

Let $(v,u)$ be a query, with $u$ an ancestor of $v$. The number of edges on
the query path is
\[
\ell=\depth(v)-\depth(u).
\]
If $\ell=0$, return $\mathrm{EMPTY}$. Otherwise set
\[
t=\mathrm{lg}[\ell],\qquad p=\mathrm{up}[t][v]=\pre_{2^t}(v).
\]
The half-open path $[v,u)$ consists of $[v,p)$ and $[p,u)$.
The minimum of $[v,p)$ is $a=\mathrm{minBlock}[t][v]$. If $p=u$, return $e(a)$.

Assume now that $p\ne u$. The remaining distance from $p$ to $u$ is
$\ell-2^t<2^t$. Since $p$ has the descendant $v$ at distance $2^t$,
Lemma~\ref{lem:ladder} implies that the base ladder of $p$ contains the
whole segment $[p,u]$. Let $z$ be the deepest vertex of this base ladder. Since
the ladder array is stored from $z$ toward the root, the positions of $p$ and
$u$ in this array are
\[
\depth(z)-\depth(p)\quad\text{and}\quad \depth(z)-\depth(u),
\]
respectively. These values are computed in $O(1)$ time from the stored pointer
to the base ladder of $p$. The modified RMQ structure on this ladder returns,
without oracle calls, the minimum $b$ on the segment $[p,u)$.

It remains to choose the smaller of $a$ and $b$ without a new comparison. The
node $b$ is an ancestor of $a$. Put
\[
d=\depth(a)-\depth(b).
\]
If $\lowerDist(a)>d$, then no ancestor of $a$ up to and including $b$ is smaller
than $a$, so $a$ is smaller than $b$ and we return $e(a)$. If
$\lowerDist(a)\le d$, the nearest smaller ancestor $c$ of $a$ lies
between $a$ and $b$. Since $a$ is the minimum on $[v,p)$, this witness
belongs to $[p,b]$. Then $b$ is no larger than $c$ in the total order.
Thus $b$ is smaller than $a$, and we return $e(b)$.
The query only reads the distance to $c$.

\section{Analysis}\label{sec:analysis}

\subsection{Precomputation}

\begin{itemize}
\item Steps 1--3 take $O(n)$ time and space without oracle calls.
\item Step 4 uses $O(n\log h)$ time, space, and oracle calls. The ladders have
total length at most $2n$, and each ladder has length at most $h+1$.
\item Step 5 builds $O(\log h)$ table entries per vertex. Each entry takes
constant time and at most two oracle calls.
\item Step 6 uses $O(n\log h)$ time and oracle calls and stores one entry per
non-root vertex.
\end{itemize}
The total preprocessing time, space, and number of oracle calls are
$O(n\log h)$ for $h\ge 2$.

\subsection{Querying}

Each query uses a constant number of table lookups, arithmetic operations,
and comparisons of indices and distances. No query invokes $\Comp$ or
$\Oracle$. Thus the cost triple is $(O(n\log h),0,O(1))$.

\section{Preprocessing Lower Bound}\label{sec:lower-bound}

For a fixed topology, a deterministic data structure that makes no oracle
calls during queries determines every answer from the preprocessing
comparison outcomes. Two assignments of distinct weights must therefore
lead to different leaves of the preprocessing decision tree whenever some
query distinguishes them.

The construction uses a branch consisting of a path $b_1,\ldots,b_X$
below a common root $r$. The node $b_1$ represents the edge $(r,b_1)$ and is
included in queries whose endpoint is $r$. Write $w(b_i)$ for the weight of
the incoming edge of $b_i$ and fix
\[
w(b_1)<w(b_2)<\cdots<w(b_X).
\]
Attach a leaf $\ell_i$ to each $b_i$. Choose the weight of its incoming edge
from any of the $i$ open intervals below $w(b_i)$ delimited by
$w(b_1),\ldots,w(b_i)$. This gives $i$ choices for its nearest smaller
ancestor: none, or one of $b_1,\ldots,b_{i-1}$. Thus the branch admits
\[
\prod_{i=1}^{X} i=X!
\]
different vectors of nearest smaller ancestors, all realizable with distinct
edge weights.

To recover the entry for $\ell_i$, move the upper query endpoint from
$\mathrm{parent}(\ell_i)$ toward $r$. If its nearest smaller ancestor is
$b_j$, the first query that returns an edge other than $e(\ell_i)$ has upper
endpoint $\mathrm{parent}(b_j)$ and returns $e(b_j)$. If every query returns
$e(\ell_i)$, there is no smaller ancestor. Thus distinct vectors give
distinct query answers.

For $2\le h<8$, use $\lfloor(n-1)/2\rfloor$ independent height-two
branches under a common root $r$. In each branch take a path $r,b,\ell$.
The two orders $w(\ell)<w(b)$ and $w(b)<w(\ell)$ are distinguished by
the query from $\ell$ to $r$. If needed, add one leaf at the root to
reach $n$ vertices. The branches give $2^{\Theta(n)}$ answer vectors and
hence an $\Omega(n\log h)$ comparison lower bound for these heights.

For $8\le h\le n-1$, set
\[
X=\left\lfloor\min\{h/2,n/4\}\right\rfloor,\qquad
q=\left\lfloor\frac{n-1}{2X}\right\rfloor.
\]
Take $q$ copies of the branch with $X$ leaves under the same root.
Each copy uses $2X$ non-root vertices, so the construction has at most
$1+2qX\le n$ vertices and height $X+1\le h$. Moreover, $qX=\Theta(n)$, since
\[
\frac n4-\frac12\le \frac{n-1}{2}-X\le qX\le\frac n2.
\]
Also $X=\Theta(h)$ because $h\le n-1$, so $\log X=\Theta(\log h)$.

Choose disjoint weight ranges for the copies and vary their assignments
independently. Add any remaining vertices as leaves adjacent to the root
with distinct fixed weights. The resulting $n$-vertex tree has at least
$(X!)^q$ distinguishable answer vectors. A binary decision tree of depth
$K$ has at most $2^K$ leaves, so the number $K$ of comparisons in the
worst case satisfies
\[
K\ge q\log_2(X!)=\Omega(qX\log X)=\Omega(n\log h).
\]
The first asymptotic bound follows by summing the logarithms of the largest
$\lceil X/2\rceil$ factors in $X!$.

\paragraph{Zero-Error Randomization.}
The same lower bound holds for worst-case expected preprocessing
comparisons of a Las Vegas structure with zero query calls. Choose one
weight assignment per distinguishable answer vector and use the uniform
distribution on this finite family. For almost every random tape, the
resulting scheme gives correct answers on the whole family. Its binary comparison
transcripts are prefix-free and distinct, so their average length is at
least the logarithm of the family size. Averaging over random tapes gives
an input with expected cost at least this large.

\paragraph{Arbitrary Endpoints.}
A zero-call path-maximum oracle gives a one-comparison online MST verifier,
so Pettie's one-comparison lower bound~\cite[Thm.~4.2]{pettie2006}
implies an $\Omega(n\log n)$ comparison bound for arbitrary paths.
The bound already holds for leaf-to-leaf queries at height two. Take
$m=\lfloor(n-1)/2\rfloor$ children of a root, each with one leaf child,
and pad to $n$ vertices by adding a leaf if needed. Make all pendant
edges heavier than the $m$ variable root edges. A query between two
of the $m$ depth-two leaves selects the lighter root edge on its path.
Thus the queries distinguish all $m!$ weight orders, which costs at least
$\log_2(m!)=\Omega(n\log n)$ preprocessing comparisons. Yang's
sorting-based Cartesian-tree structure~\cite[Sect.~2.2]{yang2021}
attains that comparison bound.

\section{Applications}\label{sec:applications}

\begin{corollary}\label{cor:yang}
For each fixed $j\ge 1$, minimum-edge queries on arbitrary tree paths admit
a data structure with
\[
(P,k,Q)=\bigl(O(n\log\lambda_j(n)),\,2j,\,O(j)\bigr)
\]
and $O(n\log\lambda_j(n))$ words of space.
\end{corollary}

\begin{proof}
Use the balanced Bor\r{u}vka tree of Yang~\cite[Sect.~3.1]{yang2021}.
It uses $N=O(n)$ vertices, height $H=O(\log n)$, and bounded degree.
Each original vertex maps to a leaf at depth $H$. Apply the
path-minimum-preserving transformation to the total order on original
edges by weight and then identifier, carrying an original edge identifier
with each stored original weight. Any dummy edge used in degree reduction
has weight $+\infty$ and can be compared directly. Thus a nonempty original
path retains its minimum edge identifier.
A merger of two distinct original edges uses one oracle call. Among
equal weights the smaller identifier wins, so ask whether the edge
with larger identifier has strictly smaller weight.

Partition depths into closed bands with boundaries $H-r_i$,
where $r_0=H$ and $r_m=0$. Write $M_i$ for the number of vertices in
band $i$ and $d_i=r_{i-1}-r_i$ for its height. Each vertex with $r$
levels below it has at least $2^r$ leaf descendants. Summing the depth
counts gives $M_i\le 2N\,2^{-r_i}$. Each band component retains bounded
degree, at least two children per internal vertex, and all leaves at
the bottom boundary. Build an all-path structure in each component.
For $j=1$, use Yang's basic zero-call structure
in $O(M_i(d_i+1))$ time and space~\cite[Sect.~3.2]{yang2021}. While
$r_{i-1}>32$, take $r_i=\lceil 2\log_2 r_{i-1}\rceil$. Finish with
one band when $r_{i-1}\le 32$. Two iterations of the doubled logarithm
are bounded by one ordinary logarithm above a fixed threshold, so
$m=O(\log^* H)=O(\lambda_1(n))$. The cost of a band before the final
one is $O(N2^{-r_i/2})$, since $r_{i-1}\le 2^{r_i/2}$.

For $j\ge 2$, take $r_i=\lambda_{j-1}(r_{i-1})$ while $r_{i-1}>32$.
Then $m=O(\lambda_j(n))$ by Yang's threshold recursion~\cite[Sect.~3.4]{yang2021}.
Use his stage-$j-1$ structure within such a band. Its preprocessing
cost is $O(M_i\lambda_{j-1}(d_i+1))=O(Nr_i2^{-r_i})$, since
$d_i+1\le r_{i-1}$. The final band has bounded height and uses
the basic structure. In both cases the positive $r_i$ decrease
strictly, so summing the band costs gives $O(N)$ time and space.

Keep the boundary vertices in an auxiliary tree. A boundary edge
stores the minimum on the path to the preceding boundary. A traversal
of the bands computes these minima in $O(N)$ time, and the auxiliary
tree has height $m$. Apply Theorem~\ref{thm:main} to this tree using
$O(N\log(m+1))$ preprocessing time and space.

For distinct query leaves, let $z$ be their lowest common ancestor,
at depth $d$. A query with coincident endpoints returns $\mathrm{EMPTY}$.
A table indexed by $d$ gives the smallest boundary depth at least $d$.
When $z$ lies on a positive-depth boundary, choose the band ending there.
Two level-ancestor queries locate the ancestors of the leaves at the
selected boundary. If $z$ lies on any boundary, both ancestors are $z$
and the in-band path is empty. This also covers the root boundary,
where the two upward paths end at the root.

Navigation takes constant time after linear
preprocessing~\cite{bender2000,bender2004}. Each upward minimum takes
constant time and zero comparisons in the auxiliary tree. A nonempty
in-band query uses at most
$2(j-1)$ comparisons and $O(j)$ time. Combining up to three minima takes
at most two further comparisons. Since $m=O(\lambda_j(n))$, the bounds
follow.
\end{proof}

\begin{corollary}\label{cor:verification}
A spanning tree can be preprocessed in $O(n\log\log n)$ time and space.
Online MST verification then takes $O(1)$ worst-case time and at most
two comparisons per query, including comparisons with the arriving edge.
\end{corollary}

\begin{proof}
Apply King's path-maximum-preserving Bor\r{u}vka transformation to the
candidate spanning tree in linear time~\cite{king1997}. Its height is
$O(\log n)$. Build a constant-time LCA structure in linear time~\cite{bender2000}.
For a query, find the lowest common ancestor of the two corresponding
leaves and obtain the maxima on the two upward paths using
Theorem~\ref{thm:main} with the order reversed. A comparison selects
the larger of the two maxima when both paths are nonempty. A second
comparison tests the weight of the arriving edge against that maximum.
Pettie used King's transformation to obtain the optimal
$O(n\log\log n)$ preprocessing time for two-comparison
verification~\cite[Sect.~4]{pettie2006}. Using Theorem~\ref{thm:main} in
this reduction also bounds the space by $O(n\log\log n)$ words and the
worst-case word-RAM query time by $O(1)$.
\end{proof}

\begin{credits}
\subsubsection{\ackname}
We thank Pavel Sokolov and Dmitry Sluch for valuable comments. Support from
the Basic Research Program of HSE University is gratefully acknowledged.

\subsubsection{\discintname}
The authors have no competing interests to declare.
\end{credits}

\bibliographystyle{splncs04}
\bibliography{references}

@techreport{alon1987,
  author      = {Alon, Noga and Schieber, Baruch},
  title       = {Optimal Preprocessing for Answering On-Line Product Queries},
  institution = {Tel Aviv University},
  number      = {71/87},
  year        = {1987},
  url         = {https://arxiv.org/abs/2406.06321}
}

@article{bender2004,
  author  = {Bender, Michael A. and Farach-Colton, Mart{\'i}n},
  title   = {The Level Ancestor Problem Simplified},
  journal = {Theoretical Computer Science},
  volume  = {321},
  number  = {1},
  pages   = {5--12},
  year    = {2004},
  doi     = {10.1016/j.tcs.2003.05.002}
}

@inproceedings{kahalon2022,
  author    = {Kahalon, Omri and Le, Hung and Milenkovi{\'c}, Lazar and Solomon, Shay},
  title     = {Can't See the Forest for the Trees: Navigating Metric Spaces by Bounded Hop-Diameter Spanners},
  booktitle = {Proceedings of the 2022 ACM Symposium on Principles of Distributed Computing},
  pages     = {151--162},
  publisher = {ACM},
  year      = {2022},
  doi       = {10.1145/3519270.3538414}
}

@article{king1997,
  author  = {King, Valerie},
  title   = {A Simpler Minimum Spanning Tree Verification Algorithm},
  journal = {Algorithmica},
  volume  = {18},
  pages   = {263--270},
  year    = {1997},
  doi     = {10.1007/BF02526037}
}

@article{komlos1985,
  author  = {Koml{\'o}s, J{\'a}nos},
  title   = {Linear Verification for Spanning Trees},
  journal = {Combinatorica},
  volume  = {5},
  number  = {1},
  pages   = {57--65},
  year    = {1985},
  doi     = {10.1007/BF02579443}
}

@article{pettie2006,
  author  = {Pettie, Seth},
  title   = {An Inverse-{Ackermann} Type Lower Bound for Online Minimum Spanning Tree Verification},
  journal = {Combinatorica},
  volume  = {26},
  number  = {2},
  pages   = {207--230},
  year    = {2006},
  doi     = {10.1007/s00493-006-0014-1}
}

@misc{yang2021,
  author       = {Yang, Tianqi},
  title        = {Tree Path Minimum Query Oracle via {Bor{\r{u}}vka} Trees},
  howpublished = {arXiv:2105.01864 [cs.DS]},
  year         = {2021},
  url          = {https://arxiv.org/abs/2105.01864}
}

@article{chazelle1987,
  author  = {Chazelle, Bernard},
  title   = {Computing on a Free Tree via Complexity-Preserving Mappings},
  journal = {Algorithmica},
  volume  = {2},
  pages   = {337--361},
  year    = {1987},
  doi     = {10.1007/BF01840366}
}

@inproceedings{bender2000,
  author    = {Bender, Michael A. and Farach-Colton, Mart{\'i}n},
  title     = {The {LCA} Problem Revisited},
  booktitle = {{LATIN} 2000: Theoretical Informatics},
  series    = {Lecture Notes in Computer Science},
  volume    = {1776},
  pages     = {88--94},
  publisher = {Springer},
  year      = {2000},
  doi       = {10.1007/10719839_9}
}

\end{document}